\documentclass{easychair}

\usepackage{xparse}

\usepackage{amsmath}
\usepackage{mathtools}
\usepackage{hyperref}
\usepackage{multicol}
\usepackage{ebproof}
\usepackage{cleveref}

\usepackage[style=alphabetic,maxnames=5]{biblatex}
\usepackage[utf8]{inputenc}
\usepackage[T1]{fontenc}
\usepackage{amsfonts}
\usepackage{amssymb}
\usepackage{cmll}
\usepackage{stmaryrd}
\usepackage{alphabeta}

\newcommand{\pdflambda}{\texorpdfstring{$λ$}{lambda}}

\usepackage{amsthm}
\newtheorem{theorem}{Theorem}[section]
\newtheorem{claim}[theorem]{Claim}
\newtheorem{corollary}[theorem]{Corollary}
\newtheorem{lemma}[theorem]{Lemma}

\newcommand{\recdef}{\Coloneqq}

\newcommand{\definitive}[1]{\textbf{#1}}

\NewDocumentCommand{\pars}{ O{} m }{#1(#2#1)}
\NewDocumentCommand{\vpars}{ m m m }
	{\IfBooleanTF{#1}{#2(#3#2)}{#3}}

\NewDocumentCommand{\braces}{ O{} m }{#1\{#2#1\}}

\NewDocumentCommand{\brackets}{ O{} m }{#1[#2#1]}

\NewDocumentCommand{\angles}{ O{} m }{#1\langle#2#1\rangle}

\NewDocumentCommand{\vbars}{ O{} m }{#1|#2#1|}

\newcommand\R{\mathsf{R}}
\newcommand\N{\mathbf{N}}

\newcommand{\LTerms}{\Lambda}
\newcommand{\RTerms}{\mathsf{L}_{\R}}
\newcommand{\STerms}{\Delta_{\R}}
\newcommand{\ATerms}{\R\angles{\Delta_{\R}}}

\NewDocumentCommand{\subst}{ s O{} m m O{} m }
    {\vpars{#1}{#2}{#3}\brackets[#5]{#6/#4}}
\NewDocumentCommand{\appl}{ s O{} m s O{} m }
	{\pars[#2]{#3}\,\vpars{#4}{#5}{#6}}
\NewDocumentCommand{\labs}{ m s O{} m }
	{\lambda{#1}.\vpars{#2}{#3}{#4}}

\NewDocumentCommand{\classOf}{ m }{\underline{#1}}

\newcommand{\zero}{\mathbf 0}
\newcommand{\sm}{.}

\NewDocumentCommand{\support}{ O{}m }{\mathsf{Supp}\pars[#1]{#2}}
\NewDocumentCommand{\LTermsOf}{ O{}m }{\LTerms\pars[#1]{#2}}
\NewDocumentCommand{\full}{osO{}m}{\vpars{#2}{#3}{#4}\mathord\downarrow\IfNoValueF{#1}{^{#1}}}

\newcommand{\bred}{\to_\Lambda}
\newcommand{\aeq}{\leftrightarrow}
\newcommand{\beq}{\leftrightarrow_\Lambda}
\newcommand{\sred}{\to}
\newcommand{\ared}{\mathrel{\widetilde{\to}}}
\newcommand{\aredRT}{\ared^*}
\newcommand{\bredRT}{\bred^*}

\newcommand{\algeq}{\triangleq}

\NewDocumentCommand{\sdesc}{}{\vdash}
\NewDocumentCommand{\adesc}{}{\Vdash}

\newcommand{\ruleVar}{(v)}
\newcommand{\ruleAbs}{($λ$)}
\newcommand{\ruleAbsA}{($λ$')}
\newcommand{\ruleApp}{(a)}
\newcommand{\ruleAppA}{(a')}
\newcommand{\ruleZ}{($\zero$)}
\newcommand{\ruleSimple}{(s)}

\newcommand{\ruleSum}{($+$)}
\newcommand{\ruleSumA}{($+$')}

\newcommand{\coefA}{a}

\newcommand{\varA}{x}
\newcommand{\varB}{y}
\newcommand{\varC}{z}

\newcommand{\ltermA}{M}
\newcommand{\ltermB}{N}
\newcommand{\ltermC}{P}

\newcommand{\atermA}{σ}
\newcommand{\atermB}{τ}
\newcommand{\atermC}{ρ}

\newcommand{\rtermA}{\ltermA}
\newcommand{\rtermB}{\ltermB}
\newcommand{\rtermC}{\ltermC}

\title{The algebraic \pdflambda-calculus is a conservative extension of the ordinary \pdflambda-calculus}

\author{
    Axel Kerinec\inst{1}
\and
    Lionel Vaux Auclair\inst{2}\thanks{
        This work was partially supported by the French ANR project
        PPS (ANR-19-CE48-0014).
    }
}

\institute{
    Université Sorbonne Paris Nord, LIPN, CNRS UMR 7030, France
\and
    Aix-Marseille Université, CNRS, I2M, France
}

\authorrunning{Kerinec, Vaux Auclair}

\titlerunning{The algebraic \pdflambda-calculus is a conservative extension of the ordinary \pdflambda-calculus}

\begin{document}

\maketitle

\begin{abstract}
    The algebraic $λ$-calculus is an extension of the ordinary $λ$-calculus with
    linear combinations of terms.
    We establish that two ordinary $λ$-terms are equivalent in the algebraic
    $λ$-calculus iff they are $β$-equal.
    Although this result was originally stated in the early 2000's (in the
    setting of Ehrhard and Regnier's differential $λ$-calculus),
    the previously proposed proofs were wrong:
    we explain why previous approaches failed and
    develop a new proof technique to establish conservativity.
\end{abstract}

\section{Introduction}

The algebraic $λ$-calculus was introduced by the second author
\cite{vaux:alglam-rta,vaux:alglam} as a generic framework to study the
rewriting theory of the $λ$-calculus in presence of weighted superpositions of
terms.
The latter feature is pervasive in the quantitative semantics of
$λ$-calculus and linear logic, that have flourished in the past twenty years
\cite[\emph{etc.}]{
ehrhard:fs,lmmp:weighted,de:pcs,ccpw:concurrentppcf}
and the algebraic $λ$-calculus is meant as a unifying syntactic counterpart
of that body of works.

The algebraic $λ$-calculus was actually obtained by removing the
differentiation primitives from Ehrhard and Regnier's differential
$λ$-calculus~\cite{er:difflamb}, keeping only the dynamics associated with
linear combinations of terms.
This dynamics is surprisingly subtle in itself:
for instance, if $1$ has an opposite in the semiring $\R$ of coefficients,
then the rewriting theory becomes trivial.
We refer the reader to the original paper \cite{vaux:alglam} for a thorough
discussion, and focus on the question of conservativity only, assuming $\R$ is
\definitive{positive} --- \emph{i.e.} if $a+b=0$ then $a=b=0$.
We briefly outline the main definitions, keeping the same notations
as in the former paper, so that the reader can consistently refer to it
for a more detailed account if need be.

\paragraph{Overview of the algebraic $λ$-calculus.}
The syntax of algebraic $λ$-terms is constructed in two stages.
We first consider \definitive{raw terms}, which are terms
inductively generated as follows:
\[ 
    \RTerms\ni
    \rtermA, \rtermB,\dotsc \recdef \varA
    \mid \labs{\varA}{\rtermA}
    \mid \appl{\rtermA}{\rtermB}
    \mid \zero
    \mid \rtermA+\rtermB
    \mid \coefA\sm\rtermA
\]
where $\coefA$ ranges over the semiring $\R$
(beware that we use Krivine's convention for application).
We consider raw terms up to $α$-equivalence: $\RTerms$ contains the set
$\LTerms$ of pure $λ$-terms as a strict subset.
We then consider \definitive{algebraic equality} $\algeq$ on raw terms,
which is the congruence generated by the equations of $\R$-module, plus the
following linearity axioms:
\begin{align*}
    \labs{\varA}{\zero}
    &\algeq \zero 
    &
    \labs{\varA}*{\rtermA+\rtermB}
    &\algeq \labs{\varA}{\rtermA}+\labs{\varA}{\rtermB}
    &
    \labs{\varA}*{\coefA\sm\rtermA}
    &\algeq \coefA\sm\labs{\varA}{\rtermA}
    \\
    \appl{\zero}\rtermC
    &\algeq \zero 
    &
    \appl*{\rtermA+\rtermB}\rtermC
    &\algeq \appl{\rtermA}\rtermC+\appl{\rtermB}\rtermC
    &
    \appl{\coefA\sm\rtermA}\rtermC
    &\algeq \coefA\sm\appl{\rtermA}\rtermC
\end{align*}
which reflects the point-wise definition of the sum of functions.
Note that, without these equations, a term such as 
$\appl{\labs{\varA}\rtermA+\labs{\varA}\rtermB}{\rtermC}$ has no redex.

The terms of the algebraic $λ$-calculus, called \definitive{algebraic terms} below,
are then the $\algeq$-classes $\atermA=\classOf{\rtermA}$ of raw terms $\rtermA\in\RTerms$.
We extend syntactic constructs to algebraic terms (e.g.,
$\labs{\varA}{\classOf{\ltermA}}=\classOf{\labs{\varA}\ltermA}$,
which is well defined because $\algeq$ is a congruence).
Among algebraic terms, we distinguish the \definitive{simple terms}, which are
intuitively those without sums at top level:
a term $\atermA$ is simple if 
$\atermA=\classOf{\varA}$ for some variable $\varA$,
or $\atermA=\labs{\varA}{\atermB}$ or $\atermA=\appl{\atermB}{\atermC}$
where $\atermB$ is itself simple (inductively).
In particular, $\classOf{\ltermA}$ is simple as soon as $\ltermA\in\LTerms$.
By definition, algebraic terms form an $\R$-module,
and it is easy to check that it is freely generated by the
set $\STerms$ of simple terms:
we write $\ATerms$ for the module of algebraic terms.

A seemingly natural way to extend the $β$-reduction $\bred$ of $λ$-terms 
to algebraic terms is to define it contextually on 
raw terms, and then apply it \emph{modulo} $\algeq$:
among other issues with this naïve definition,
note that $M\algeq M+0\sm N$ would reduce to $M+0\sm N'\algeq M$
for any $N\bred N'$, so that the obtained reduction 
would be reflexive and there would be no clear notion of normal form.
Ehrhard and Regnier's solution is to rather consider two relations:
$\mathord{\sred}\subset \STerms\times\ATerms$ 
defined contextually on simple terms with $β$-reduction
as a base case;
and $\mathord{\ared}\subset \ATerms\times\ATerms$ 
on algebraic terms, obtained by setting $σ\ared σ'$ iff $σ=a\sm τ+ρ$
and $σ'=a\sm τ'+ρ$ with $τ\sred τ'$ and $a\not=0$.
Then $\ared$ is confluent \cite{er:difflamb,vaux:alglam} and, \emph{provided
$\R$ is positive}, an algebraic term is in normal form iff it is the class of a
raw term without $β$-redex.

Note that, for any fixed point combinator $Y$,
by setting $\infty_{σ}=\appl{\classOf{Y}}{\labs{\varA}*{σ+\classOf{\varA}}}$,
we obtain $\infty_{σ}\aeq σ+\infty_{σ}$ where
$\aeq$ is the equivalence on $\ATerms$ generated by the reduction relation $\ared$.
In case $1\in\R$ has an opposite $-1$,
we can now exhibit the above-mentioned inconsistency of the theory:
$\classOf{\zero}=\infty_{σ}+(-1)\sm\infty_{σ}\aeq σ$ for any $σ$.
From now on, we thus assume that $\R$ is positive.

\paragraph{Contributions.}
Our goal is to establish that, for any two $λ$-terms $\ltermA$ and $\ltermB\in\LTerms$,
we have $\classOf{\ltermA}\aeq\classOf{\ltermB}$ iff $\ltermA\beq\ltermB$,
where $\beq$ is the usual $β$-equivalence on $λ$-terms.
For that purpose, it is sufficient to establish a conservativity result on
reduction relations rather than on the induced equivalences:
if $\classOf{\ltermA}\aredRT\classOf{\ltermB}$ then $\ltermA\bredRT\ltermB$.
This is our main result, \cref{theorem:conservativity} below.

In the next section, we explain what was wrong with the previous two attempts,
first by Ehrhard and Regnier, then by the second author, to establish
conservativity, and we outline the new proof strategy we propose.
The rest of the paper is dedicated to the proof of \cref{theorem:conservativity}.%
\footnote{
These results were obtained during a research internship of the first author,
in the first half of 2019; they were presented by the second author at the
annual meeting of the working group Scalp (\emph{Structures formelles pour le
Calcul et les Preuves}) in Lyon in October 2019.
This collaboration was unfortunately disrupted by the COVID-19 pandemic in the
following year, which delayed dissemination to a wider audience.
}

\section{Two non-proofs and a new approach}

Recall that an \definitive{ARS} (abstract rewriting system) is a pair
$(A,\leadsto)$ of a set $A$ and binary relation
$\mathord{\leadsto}\subseteq A\times A$.
An \definitive{extension} of $(A,\leadsto)$ is another ARS $(A',\leadsto')$
such that $A\subseteq A'$ and $\mathord{\leadsto}\subseteq\mathord{\leadsto'}$.
This extension is \definitive{conservative} if, for every $a_1,a_2\in A$,
$a_1\leadsto a_2$ iff $a_1\leadsto' a_2$.
An \definitive{equational system} is an ARS $(A,\sim)$ such that $\sim$
is an equivalence relation.
Our goal is thus to establish that the equational system
$(\ATerms,\aeq)$ is a conservative extension of $(\LTerms,\beq)$
--- here we consider the injection
$\ltermA\in\LTerms\mapsto\classOf{\ltermA}\in\ATerms$
as an inclusion.

In their paper on the differential $λ$-calculus \cite{er:difflamb},
Ehrhard and Regnier claim that this follows directly from the confluence of
$\ared$, but this argument is not valid:
$\ared$ does contain $\bred$, and it is indeed confluent,
without any positivity assumption;
but we have already stated that 
$\aeq$ is inconsistent in presence of negative coefficients,
so this observation cannot be sufficient.

Ehrhard and Regnier's mistake is certainly an erroneous application
of a general conservativity result in Terese's textbook \cite{terese},
missing the fact that Terese's notion of extension is more demanding:
for Terese, $(A,\leadsto)$ is a \definitive{sub-ARS} of $(A',\leadsto')$ if 
$A\subseteq A'$ and, for every $a\in A$ and $a'\in A'$,
$a\leadsto' a'$ iff $a'\in A$ and $a\leadsto a'$.
The latter is strictly stronger than the mere inclusion
$\mathord{\leadsto}\subseteq\mathord{\leadsto'}$,
and is indeed sufficient to deduce conservativity
for the induced equational systems from the confluence of the super-ARS
\cite[Exercice 1.3.21 (iii)]{terese}.
But $(\LTerms,\bred)$ is not a sub-ARS of $(\ATerms,\ared)$,
even when $\R$ is positive:
for instance, if $\R=\mathbf Q^+$ and $\ltermA\bred \ltermA'\not=\ltermA$,
we have $\classOf{\ltermA}=\frac 12\classOf{\ltermA}+\frac 12\classOf{\ltermA}
\ared \frac 12\classOf{\ltermA}+\frac 12\classOf{\ltermA'}\not\in\classOf{\LTerms}$.
So one must design another approach.

Given $\atermA\in\ATerms$, one can consider the
finite set of $λ$-terms $\LTermsOf{\atermA}\subset\LTerms$
obtained by keeping exactly one element
in the support of each sum occurring in $\atermA$
\cite[Definition 3.18]{vaux:alglam}.
The second author tried to establish the conservativity of $\aredRT$ 
over $\bredRT$ by iterating the following:
\begin{claim}[{\cite[Lemma 3.20]{vaux:alglam}}]
    \label{claim:wrong}
    If $\atermA\ared\atermA'$ and $\ltermA'\in\LTermsOf{\atermA'}$
    then there exists $\ltermA\in\LTermsOf{\atermA}$ such that
    $\ltermA\bredRT\ltermA'$.
\end{claim}
But the latter claim is wrong!
Consider, for instance, 
$\atermA
=\appl*{\labs{\varA}{\appl{\classOf{\varA}}{\classOf{\varA}}}}*{\classOf{\varB}+\classOf{\varC}}
\ared\atermA'
=\appl*{\classOf{\varB}+\classOf{\varC}}*{\classOf{\varB}+\classOf{\varC}}
$.
We have $\ltermA'=\appl{{\varB}}{{\varC}}\in\LTermsOf{\atermA'}$
but no term in $\LTermsOf{\atermA'}=
\{
\appl*{\labs{\varA}{\appl{{\varA}}{{\varA}}}}{{\varB}},
\appl*{\labs{\varA}{\appl{{\varA}}{{\varA}}}}{{\varC}}
\}
$ $β$-reduces to $\ltermA'$.
Note that, in this counter-example, there is no $\ltermA\in\LTerms$
such that $\classOf{\ltermA}\aredRT\atermA$:
somehow, we must exploit this additional hypothesis
to establish a correct version of \cref{claim:wrong}.

Reasoning on $\aredRT$ directly is difficult,
due to its definition as a reflexive and transitive closure.
The technique we propose involves the definition of a mixed-type relation $\ltermA \adesc \atermA$
between a $λ$-term $\ltermA$ and a term $\atermA\in\ATerms$:
intuitively, $\ltermA\adesc\atermA$ when $\atermA$ is obtained by pasting
together terms issued from various reductions of $\ltermA$,
and we say $\ltermA\adesc\atermA$ is a \definitive{mashup} of
such reductions.
In particular: $\ltermA\adesc\classOf{\ltermA'}$ as soon as 
$\ltermA\bredRT\ltermA'$;
and $\ltermA\adesc\atermA+\atermB$ as soon as
$\ltermA\adesc\atermA$ and $\ltermA\adesc\atermB$.
We then show that $\adesc$ is conservative over $\bredRT$
(\cref{lemma:lterm}) and that $\ltermA\adesc\atermA$ 
as soon as $\classOf{\ltermA}\aredRT\atermA$
(\cref{lemma:reflexivity,lemma:desc:ared}):
this ensures the conservativity of $\aredRT$ over $\bredRT$
(\cref{theorem:conservativity}).
Our whole approach thus rests on the careful definition of the mashup relation.
Among other requirements, it must behave well w.r.t. the structure of terms:
e.g., if $\ltermA\adesc\atermA$ then $\labs{\varA}{\ltermA}\adesc\labs{\varA}{\atermA}$.

\section{Mashup of $β$-reductions}

We define two relations
$\mathord{\sdesc}\subseteq\LTerms\times\STerms$
and 
$\mathord{\adesc}\subseteq\LTerms\times\ATerms$
by mutual induction, with the rules of \cref{fig:adesc}.
If $σ\in\ATerms$, we write $\support{σ}\subset\STerms$
for its support set.

\begin{figure}[t]
\begin{gather*}
   \begin{prooftree}
        \hypo{\ltermA \bredRT\varA}
        \infer1[\ruleVar]{\ltermA\sdesc\classOf{\varA}}
    \end{prooftree}
        \qquad
    \begin{prooftree}
        \hypo{\ltermA \bredRT\labs{\varA}{\ltermB}}
        \hypo{\ltermB \sdesc \atermB} 
        \infer2[\ruleAbs]{\ltermA\sdesc\labs{\varA}{\atermB}}
    \end{prooftree}
    \qquad
    \begin{prooftree}
        \hypo{\ltermA \bredRT \appl{\ltermB}\ltermC}
        \hypo{\ltermB \sdesc \atermB}
        \hypo{\ltermC \adesc \atermC}
        \infer3[\ruleApp]{\ltermA\sdesc\appl{\atermB}{\atermC}}
    \end{prooftree}
     \\[1em]
    \begin{prooftree}
        \infer0[\ruleZ]{\ltermA\adesc\zero}
    \end{prooftree}
     \qquad
    \begin{prooftree}
    \hypo{\ltermA\sdesc \atermA}
    \hypo{\ltermA\adesc \atermB}
    \infer2[\ruleSum]{\ltermA\adesc \coefA\sm\atermA+\atermB}
    \end{prooftree}
\end{gather*}
\caption{Inference rules for the mashup relations}
\label{fig:adesc}
\end{figure}

\begin{figure}[t]
\begin{gather*}
 \begin{prooftree}
        \hypo{\ltermA \sdesc \atermA}
        \infer1[\ruleSimple]{\ltermA \adesc \atermA}
    \end{prooftree}
    \\[1em]
    \begin{prooftree}
        \hypo{\ltermA \bredRT\labs{\varA}{\ltermB}}
        \hypo{\ltermB \adesc \atermB} 
        \infer2[\ruleAbsA]{\ltermA\adesc\labs{\varA}{\atermB}}
    \end{prooftree}
    \qquad
    \begin{prooftree}
        \hypo{\ltermA \bredRT\appl{\ltermB}\ltermC}
        \hypo{\ltermB \adesc \atermB}
        \hypo{\ltermC \adesc \atermC}
        \infer3[\ruleAppA]{\ltermA\adesc\appl{\atermB}{\atermC}}
    \end{prooftree}
    \qquad
    \begin{prooftree}
    \hypo{\ltermA\adesc \atermA}
    \hypo{\ltermA\adesc \atermB}
    \infer2[\ruleSumA]{\ltermA\adesc \coefA\sm\atermA+\atermB}
    \end{prooftree}
\end{gather*}
\caption{Admissible rules for the mashup relations}
\label{fig:admissible}
\end{figure}

\begin{lemma}\label{lemma:support}
    We have $\ltermA\adesc\atermA$ iff, for every $\atermA'\in\support{\atermA}$,
    $\ltermA\sdesc\atermA'$.
\end{lemma}
\begin{proof}
    The forward implication is done by induction on the derivation of $\ltermA\adesc\atermA$,
    noting that if $\atermA'\in\support{a\atermB+\atermC}$ with $\ltermA\sdesc\atermB$
    and $\ltermA\adesc\atermC$ then $\atermA'=\atermB$ or $\atermA'\in\support{\atermC}$.
    For the reverse implication, we can write $\atermA=\sum_{i=1}^n\coefA_i\sm\atermA_i$
    with $\atermA_i\in\support{\atermA}$ for $1\le i\le n$,
    and obtain a derivation of $\ltermA\adesc\atermA$ by induction on $n$.
\end{proof}

\begin{lemma}\label{lemma:admissible}
    The rules of \cref{fig:admissible} are admissible.
\end{lemma}
\begin{proof}
    For \ruleSimple,
    it is sufficient to observe that $\atermA=1\atermA+\zero$.
    For the other three rules, we reason on the support sets,
    using \cref{lemma:support}.
\end{proof}

\begin{lemma}[Reflexivity of $\sdesc$]\label{lemma:reflexivity}
    For every $\ltermA\in\LTerms$, $\ltermA\sdesc\classOf{\ltermA}$.
\end{lemma}
\begin{proof}
    By a straightforward induction on $\ltermA$, using the reflexivity of $\bredRT$
    and rule \ruleSimple.
\end{proof}

\begin{lemma}[Conservativity of $\adesc$]\label{lemma:lterm}
If $\ltermA\adesc\classOf{\ltermA'}$ then $\ltermA\bredRT\ltermA'$.
\end{lemma}
\begin{proof}
    Note that $\classOf{\ltermA'}\in\STerms$, hence
    $\ltermA\sdesc\classOf{\ltermA'}$ by \cref{lemma:support}.
    The proof is then by induction on $\ltermA'$,
    inspecting the last rule of the derivation of $\ltermA\sdesc\classOf{\ltermA'}$:
    \begin{itemize}
        \item[\ruleVar]
            If $\ltermA\bredRT \varA$ and $\classOf{\ltermA'}=\classOf{\varA}$ 
            then we conclude directly since $\ltermA'=\varA$.
        \item[\ruleAbs]
            If $\ltermA\bredRT \labs{\varA}{\ltermB}$ and $\ltermB\sdesc\atermB$ 
            with $\classOf{\ltermA'}=\labs{\varA}{\atermB}$,
            then $\ltermA'=\labs{\varA}{\ltermB'}$ with $\atermB=\classOf{\ltermB'}$.
            By induction hypothesis, $\ltermB\bredRT\ltermB'$,
            hence $\ltermA\bredRT\ltermA'$.
        \item[\ruleApp]
            If $\ltermA\bredRT \appl{\ltermB}{\ltermC}$,
            $\ltermB\sdesc\atermB$ and $\ltermC\adesc\atermC$ 
            with $\classOf{\ltermA'}=\appl{\atermB}{\atermC}$,
            then $\ltermA'=\appl{\ltermB'}{\ltermC'}$
            with $\atermB=\classOf{\ltermB'}$
            and $\atermC=\classOf{\ltermC'}$.
            In particular $\atermC\in\STerms$, hence
            $\ltermC\sdesc\atermC$ by \cref{lemma:support}.
            By induction hypothesis, $\ltermB\bredRT\ltermB'$
            and $\ltermC\bredRT\ltermC'$,
            hence $\ltermA\bredRT\ltermA'$. \qedhere
    \end{itemize}
\end{proof}

\begin{lemma}[Compatibility with $\bred$]\label{lemma:bred:desc}
    If $\ltermA\bredRT\ltermA'\sdesc\atermA$ then $\ltermA\sdesc\atermA$.
    Similarly, if $\ltermA\bredRT\ltermA'\adesc\atermA$ then $\ltermA\adesc\atermA$.
\end{lemma}
\begin{proof}
    For the first implication, it is sufficient to inspect the last rule
    of the derivation $\ltermA'\sdesc\atermA$, and use the transitivity 
    of $\bredRT$.
    The second implication follows directly by induction on the derivation of 
    $\ltermA'\adesc\atermA$.
\end{proof}

\section{Conservativity of algebraic reduction}

\begin{lemma}[Substitution lemma]\label{lemma:subst}
    If $\ltermA\adesc\atermA$ and $\ltermC\adesc\atermC$
    then $\subst{\ltermA}{\varA}{\ltermC}\adesc\subst{\atermA}{\varA}{\atermC}$.
\end{lemma}
\begin{proof}
    We prove the result, together with the variant assuming 
    $\ltermA\sdesc\atermA$ instead of $\ltermA\adesc\atermA$,
    by induction on the derivations of those judgements.
    \begin{itemize}
        \item[\ruleVar]
            If $\ltermA\bredRT \varB$ and $\atermA=\classOf{\varB}$ then:
            \begin{itemize}
                \item if $\varB=\varA$, then
                    $\subst{\ltermA}{\varA}{\ltermC}
                    \bredRT \subst{\varA}{\varA}{\ltermC}
                    =\ltermC
                    \adesc\atermC$
                    and we obtain
                    $\subst{\ltermA}{\varA}{\ltermC}
                    \adesc\atermC
                    =\subst{\atermA}{\varA}{\atermC}$
                    by \cref{lemma:bred:desc};
                \item otherwise,
                    $\subst{\ltermA}{\varA}{\ltermC}
                    \bredRT \subst{\varB}{\varA}{\ltermC}
                    =\varB$,
                    hence
                    $\subst{\ltermA}{\varA}{\ltermC}
                    \adesc\classOf{\varB}
                    =\subst{\atermA}{\varA}{\atermC}$
                    by \ruleVar.
            \end{itemize}
        \item[\ruleAbs]
            If $\ltermA\bredRT \labs{\varB}{\ltermB}$ and $\ltermB\sdesc\atermB$ 
            with $\atermA=\labs{\varB}{\atermB}$ (choosing $\varB\not=\varA$ and 
            $\varB$ not free in $\ltermC$ nor in $\atermC$),
            then $\subst{\ltermA}{\varA}{\ltermC}
            \bredRT\labs{\varB}{\subst{\ltermB}{\varA}{\ltermC}}$
            and, by induction hypothesis
            $\subst{\ltermB}{\varA}{\ltermC}\adesc\subst{\atermB}{\varA}{\atermC}$:
            we obtain
            $\subst{\ltermA}{\varA}{\ltermC}
            \adesc\labs{\varB}{\subst{\atermB}{\varA}{\atermC}}
            =\subst{\atermA}{\varA}{\atermC}$
            by $\ruleAbsA$.
        \item[\ruleApp]
            If $\ltermA\bredRT \appl{\ltermB_1}{\ltermB_2}$,
            $\ltermB_1\sdesc\atermB_1$ and $\ltermB_2\adesc\atermB_2$,
            with $\atermA=\appl{\atermB_1}{\atermB_2}$,
            then we have $\subst{\ltermA}{\varA}{\ltermC}
            \bredRT\appl{\subst{\ltermB_1}{\varA}{\ltermC}}{\subst{\ltermB_2}{\varA}{\ltermC}}$
            and, by induction hypothesis,
            $\subst{\ltermB_1}{\varA}{\ltermC}\adesc\subst{\atermB_1}{\varA}{\atermC}$ and
            $\subst{\ltermB_2}{\varA}{\ltermC}\adesc\subst{\atermB_2}{\varA}{\atermC}$:
            we obtain
            $\subst{\ltermA}{\varA}{\ltermC}
            \adesc\appl{\subst{\atermB_1}{\varA}{\atermC}}{\subst{\atermB_2}{\varA}{\atermC}}
            =\subst{\atermA}{\varA}{\ltermC}$
            by $\ruleAppA$.
        \item[\ruleZ]
            If $\atermA=\zero$ then $\subst{\atermA}{\varA}{\atermC}=\zero$
            and we conclude directly, by \ruleZ.
        \item[\ruleSum]
            If $\atermA=\coefA\sm\atermB_1+\atermB_2$ 
            with $\ltermA\sdesc\atermB_1$ and $\ltermA\adesc\atermB_2$,
            then,
            by induction hypothesis, 
            $\subst{\ltermA}{\varA}{\ltermC}\adesc\subst{\atermB_1}{\varA}{\atermC}$ and
            $\subst{\ltermA}{\varA}{\ltermC}\adesc\subst{\atermB_2}{\varA}{\atermC}$,
            hence
            $\subst{\ltermA}{\varA}{\ltermC}
            \adesc\coefA\sm\subst{\atermB_1}{\varA}{\atermC}+\subst{\atermB_2}{\varA}{\atermC}
            =\subst{\atermA}{\varA}{\atermC}$
            by $\ruleAppA$. \qedhere
    \end{itemize}
\end{proof}

Note that, by positivity, if $\atermA=\coefA\sm\atermB+\atermC$ 
with $\atermB\in\STerms$ and $\coefA\not=0$,
then $\atermB\in\support{\atermA}\supseteq\support{\atermC}$.

\begin{lemma}[Compatibility with $\ared$]\label{lemma:desc:ared}
    Let $\ltermA\in\LTerms$ and $\atermA'\in\ATerms$.
    For every $\atermA\in\STerms$ such that 
    $\ltermA\sdesc\atermA\sred\atermA'$ 
    (resp. every $\atermA\in\ATerms$ such that 
    $\ltermA\adesc\atermA\ared\atermA'$),
    we have $\ltermA\adesc\atermA'$.
\end{lemma}
\begin{proof}
    The proof is by induction on the definition of the reduction
    $\atermA\sred\atermA'$ or $\atermA\ared\atermA'$.
    \begin{itemize}
        \item
            If $\atermA=\appl*{\labs{\varA}{\atermB}}{\atermC}\in\STerms$
            and $\atermA'=\subst{\atermB}{\varA}{\atermC}$,
            then the derivation of $\ltermA\sdesc\atermA$
            must be of the form
            \[\begin{prooftree}
                \hypo{\ltermA\bredRT\appl{\ltermB}{\ltermC}}
                \hypo{\ltermB\bredRT\labs{\varA}{\ltermB'}}
                \hypo{\ltermB'\sdesc\atermB}
                \infer2[\ruleAbs]{\ltermB\sdesc\labs{\varA}{\atermB}}
                \hypo{\ltermC\adesc\atermC}
                \infer3[\ruleApp]{\ltermA\sdesc\appl*{\labs{\varA}{\atermB}}{\atermC}}
            \end{prooftree}\;.\]
            By \cref{lemma:subst}, we have 
            $\subst{\ltermB'}{\varA}{\ltermC}\adesc\atermA'$.
            Moreover,
            $\ltermA\bredRT\appl{\ltermB}{\ltermC}
            \bredRT\appl*{\labs{\varA}{\ltermB'}}{\ltermC}
            \bred\subst{\ltermB'}{\varA}{\ltermC}$
            and we obtain
            $\ltermA\adesc\atermA'$ by \cref{lemma:bred:desc}.
        \item 
            If $\atermA=\labs{\varA}{\atermB}$
            and $\atermA'=\labs{\varA}{\atermB'}$
            with $\atermB\sred\atermB'$,
            then the derivation of $\ltermA\sdesc\atermA$
            must be of the form
            \[\begin{prooftree}
                \hypo{\ltermA\bredRT\labs{\varA}{\ltermB}}
                \hypo{\ltermB\sdesc\atermB}
                \infer2[\ruleAbs]{\ltermA\sdesc\labs{\varA}{\atermB}}
            \end{prooftree}\;.\]
            We obtain  $\ltermB\adesc\atermB'$ by induction hypothesis,
            and we conclude by \ruleAbsA.
        \item 
            If $\atermA=\appl{\atermB}{\atermC}$
            and $\atermA'=\appl{\atermB'}{\atermC'}$
            with either $\atermB\sred\atermB'$ and $\atermC=\atermC'$,
            or  $\atermB=\atermB'$ and $\atermC\ared\atermC'$,
            then the derivation of $\ltermA\sdesc\atermA$
            must be of the form
            \[\begin{prooftree}
                \hypo{\ltermA\bredRT\appl{\ltermB}{\ltermC}}
                \hypo{\ltermB\sdesc\atermB}
                \hypo{\ltermC\adesc\atermC}
                \infer3[\ruleApp]{\ltermA\sdesc\appl{\atermB}{\atermC}}
            \end{prooftree}\;.\]
            We obtain  $\ltermB\adesc\atermB'$ and  $\ltermC\adesc\atermC'$ 
            by induction hypothesis,
            and we conclude by \ruleAppA.
        \item 
            If $\atermA=\coefA\sm\atermB+\atermC$
            and $\atermA'=\coefA\sm\atermB'+\atermC$
            with $\atermB\sred\atermB'$ and $\coefA\not=0$,
            then we have already observed that 
            $\atermB\in\support{\atermA}$ and
            $\support{\atermC}\subseteq\support{\atermA}$.
            Since $\ltermA\adesc\atermA$,
            we obtain $\ltermA\sdesc\atermB$
            and $\ltermA\adesc\atermC$
            by \cref{lemma:support}.
            The induction hypothesis gives $\ltermA\adesc\atermB'$,
            hence
            $\ltermA\adesc\coefA\sm\atermB'+\atermC=\atermA'$
            by \ruleSumA. \qedhere
    \end{itemize}
\end{proof}

\begin{theorem}[Conservativity of $\aredRT$]\label{theorem:conservativity}
    If $\classOf{\ltermA}\aredRT\classOf{\ltermB}$
    then $\ltermA\bredRT\ltermB$.
\end{theorem}
\begin{proof}
    Assume $\classOf{\ltermA}\aredRT\classOf{\ltermB}$.
    By \cref{lemma:reflexivity} and \ruleSimple, we have 
    $\ltermA\adesc\classOf{\ltermA}$.
    By iterating \cref{lemma:desc:ared}, we deduce
    $\ltermA\adesc\classOf{\ltermB}$.
    We conclude by \cref{lemma:lterm}.
\end{proof}

\begin{corollary}[Conservativity of $\aeq$]
    If $\R$ is positive then $\classOf{\ltermA}\aeq\classOf{\ltermB}$
    iff $\ltermA\beq\ltermB$.
\end{corollary}
\begin{proof}
    Assuming $\classOf{\ltermA}\aeq\classOf{\ltermB}$, the confluence
    of $\ared$ ensures that there exist $\atermA\in\ATerms$ and $k\in\N$,
    such that $\classOf{\ltermA}\ared^k\atermA$ and $\classOf{\ltermB}\aredRT\atermA$.
    It follows \cite[Lemma 3.23]{vaux:alglam} that 
    $\atermA\aredRT\full[k]{\classOf{\ltermA}}$
    where $\full{\atermB}$ is the term obtained by reducing 
    all the redexes of $\atermB$ simultaneously.
    Observing that $\full{\classOf{\ltermA}}=\classOf{\full{\ltermA}}$,
    we obtain $\classOf{\ltermB}\aredRT\classOf{\full[k]{\ltermA}}$
    hence $\ltermB\bredRT\full[k]{\ltermA}$ by \cref{theorem:conservativity},
    which concludes the proof.
\end{proof}

\printbibliography

\end{document}